\documentclass[11pt]{article}
\usepackage[utf8]{inputenc}
\usepackage{style}
\usepackage{framed}
\usepackage[parfill]{parskip}
\usepackage{setspace}
\usepackage{hyperref}

\allowdisplaybreaks

\title{A Note on Max $k$-Vertex Cover: Faster FPT-AS, Smaller Approximate Kernel and Improved Approximation}

\author{
Pasin Manurangsi\thanks{UC Berkeley. Email: \texttt{pasin@berkeley.edu}. Supported by NSF under Grants No. CCF 1655215 and CCF 1815434.}
}

\begin{document}
\clearpage
\maketitle
\thispagestyle{empty}

\begin{abstract}
In Maximum $k$-Vertex Cover (Max $k$-VC), the input is an edge-weighted graph $G$ and an integer $k$, and the goal is to find a subset $S$ of $k$ vertices that maximizes the total weight of edges covered by $S$. Here we say that an edge is covered by $S$ iff at least one of its endpoints lies in $S$.

We present an FPT approximation scheme (FPT-AS) that runs in $(1/\varepsilon)^{O(k)} \poly(n)$ time for the problem, which improves upon Gupta, Lee and Li's $(k/\varepsilon)^{O(k)} \poly(n)$-time FPT-AS~\cite{GLL18a,GLL18}. Our algorithm is simple: just use brute force to find the best $k$-vertex subset among the $O(k/\varepsilon)$ vertices with maximum weighted degrees.

Our algorithm naturally yields an (efficient) approximate kernelization scheme of $O(k/\varepsilon)$ vertices; previously, an $O(k^5/\varepsilon^2)$-vertex approximate kernel is only known for the unweighted version of Max $k$-VC~\cite{LPRS17}. Interestingly, this also has an application outside of parameterized complexity: using our approximate kernelization as a preprocessing step, we can directly apply Raghavendra and Tan's SDP-based algorithm for 2SAT with cardinality constraint~\cite{RT12} to give an $0.92$-approximation algorithm for Max $k$-VC in polynomial time. This improves upon the best known polynomial time approximation algorithm of Feige and Langberg~\cite{FL01} which yields $(0.75 + \delta)$-approximation for some (small and unspecified) constant $\delta > 0$.

We also consider the minimization version of the problem (called Min $k$-VC), where the goal is to find a set $S$ of $k$ vertices that minimizes the total weight of edges covered by $S$. We provide a FPT-AS for Min $k$-VC with similar running time of $(1/\varepsilon)^{O(k)} \poly(n)$. Once again, this improves on a $(k/\varepsilon)^{O(k)} \poly(n)$-time FPT-AS of Gupta et al. On the other hand, we show, assuming a variant of the Small Set Expansion Hypothesis~\cite{RS10} and $\NP \nsubseteq \coNP/\poly$, that there is no polynomial size approximate kernelization for Min $k$-VC for any factor less than two.
\end{abstract}

\newpage
\setcounter{page}{1}
\section{Introduction}

In the \emph{Vertex Cover} problem, we are given a graph $G$ and an integer $k$, and the goal is to determine whether there is a set $S$ of $k$ vertices that covers all the edges, where the edge is said to be covered by $S$ if at least one of its endpoints lies in $S$. Vertex Cover is a classic graph problem and is among Karp's original list of 21 NP-complete problems~\cite{Karp72}. This NP-hardness has led to studies of variants of the problems. One such direction is to consider the optimization versions of the problem. Arguably, the two most natural optimization formulations of Vertex Cover are the \emph{Minimum Vertex Cover (Min VC)} problem, where the constraint that every edge is covered is treated as a hard constraint and the goal is to find $S$ with smallest size that satisfies this, and the \emph{Maximum $k$-Vertex Cover (Max $k$-VC)} problem\footnote{Max $k$-VC and Min $k$-VC (which will be introduced below) are sometimes referred to as the Max Partial Vertex Cover and Min Partial Vertex Cover respectively. However, we decide against calling them as such to avoid ambiguity since Partial Vertex Cover has also used to refer to a different variant of Vertex Cover (see e.g.~\cite{Bar-Yehuda01}).}, where the cardinality constraint $|S| = k$ is treated as a hard constraint and the goal is to find such $S$ that covers as many edges as possible.

Both problems have been thoroughly studied in the approximation algorithms and hardness of approximation literature. Min VC admits a simple greedy 2-approximation algorithm\footnote{Throughout this note, we use the convention that the approximation ratio is the worst case ratio between the cost of the output solution and the optimum. In other words, the approximation ratios for maximization problems will be at most one, whereas the approximation ratios for minimization problems will be at least one.}, which has been known since the seventies (see e.g.~\cite{GJ79}). The approximation ratio has subsequently been slightly improved~\cite{BE85,MS85} and, currently, the best known approximation ratio in polynomial time is $(2 - 1/O(\sqrt{\log n}))$~\cite{Kar09}. There has also been a number of works on hardness of approximation of Vertex Cover~\cite{BGS98,Hastad01,DS05,KR08,BK09,KMS17,KMS18}. The best known NP-hardness of approximation for Min VC, established in the recent works that resolve the (imperfect) 2-to-1 conjecture~\cite{KMS17,DKKMS16,DKKMS17,KMS18}, has a factor of $(\sqrt{2} - \varepsilon)$ for any $\varepsilon > 0$. Assuming the Unique Games Conjecture (UGC)~\cite{Kho02}, the inapproximability ratio can be improved to $(2 - \varepsilon)$ for any $\varepsilon > 0$~\cite{KR08,BK09}, which is tight up to lower order terms.

Unlike Min VC, tight approximability results for Max $k$-VC are not known (even assuming UGC). In particular, on the algorithmic front, the best known efficient approximation algorithm due to Feige and Langberg~\cite{FL01} yields a $(0.75 + \delta)$-approximation for the problem, where $\delta > 0$ is a (small) constant. This was an improvement over an earlier 0.75-approximation algorithm of Ageev and Sviridenko~\cite{AS04}, which in turn improved upon the simple greedy algorithm that yields $(1 - 1/e)$-approximation for the problem~\cite{Hoch96}. (See also~\cite{HanYZ02,HanYZZ02,HalperinZ02} where improvements have been made for certain ranges of $k$ and $n$.) On the hardness of approximation front, it is known that the problem is NP-hard to approximate to within $(1 + \delta)$ factor for some (small) $\delta > 0$~\cite{Patrank94}. Moreover, it follows from a result of Austrin, Khot and Safra~\cite{AKS11} that it is UG-hard to approximate the problem to within a factor of 0.944. (See Appendix~\ref{app:aks}.) This leaves quite a large gap between the upper and lower bounds, even assuming the UGC.

Approximability is not the only aspect of Vertex Cover and its variants that has been thoroughly explored: its parameterized complexity is also a well-studied subject. Recall that an algorithm is said to be \emph{fixed-parameter (FPT)} with respect to parameter $k$ if it runs in time $f(k) \cdot \poly(n)$ for some function $f$, where $n$ is the size of the input. An FPT algorithm (with running time $k^{O(k)}\cdot \poly(n)$) was first devised for Vertex Cover by Buss and Goldsmith~\cite{BussG93}. Since then, many different FPT algorithms have been discovered for Vertex Cover; to the best of our knowledge, the fastest known algorithm is that of Chen, Kanj and Xia~\cite{ChenKX10}, which runs in $1.2738^k \cdot \poly(n)$ time.

Notice that an FPT algorithm for Vertex Cover can also be adapted to solve Min VC in FPT time parameterized by the optimal solution size, by running the Vertex Cover algorithm for $k = 1, 2, \dots$ until it finds the size of the optimal solution. On the other hand, Max $k$-VC is unlikely to admit an FPT algorithm, as it is \W[1]-hard~\cite{GuoNW07}. Circumventing this hardness, Marx~\cite{Marx08} designed an \emph{FPT approximation scheme (FPT-AS)}, which is an FPT algorithm that can achieve approximation ratio $(1 - \varepsilon)$ (or $(1 + \varepsilon)$ for minimization problems) for any $\varepsilon > 0$, for Max $k$-VC. In particular, his algorithm runs in time $(k/\varepsilon)^{O(k^3/\varepsilon)} \cdot \poly(n)$. This should be contrasted with the aforementioned fact that Max $k$-VC does not admit a PTAS unless $\P = \NP$. Recently, the FPT-AS has been sped up by Gupta, Lee and Li~\cite{GLL18a,GLL18}\footnote{In fact, Gupta \etal gives an FPT-AS for Min $k$-VC; it is trivial to see that their algorithm works for Max $k$-VC as well.} to run in time $(k/\varepsilon)^{O(k)} \cdot \poly(n)$.

FPT algorithms are intimately connected to the notion of kernel. A \emph{kernelization algorithm} (or \emph{kernel}) of a parameterized problem is a polynomial time algorithm that, given an instance $(I, k)$, produces another instance $(I', k')$ such that the size of the new instance $|I'|$ and the new parameter $k'$ are both bounded by $g(k)$ for some function $g$. It is well known that a parameterized problem admits a kernel if and only if it admits FPT algorithms~\cite{CaiCDF97}. Once again, many kernels are known for Vertex Cover (see e.g.~\cite{Abu-KhzamCFLSS04} and references therein). On the other hand, the \W[1]-hardness of Max $k$-VC means that it does not admit a kernel unless $\W[1] = \FPT$.

Recently, there have been attempts to make the concept of kernelization compatible with approximation algorithms~\cite{FKRS18,LPRS17}. In this note, we follow the notations defined by Lokshtanov \etal~\cite{LPRS17}. For our purpose, it suffices to define an \emph{$\alpha$-approximate kernel} for an parameterized optimization problem as a pair of polynomial time algorithms $\cA$, the \emph{reduction algorithm}, and $\cB$, the \emph{solution lifting algorithm}, such that (i) given an instance $(I, k)$, $\cA$ produces another instance $(I', k')$ such that $|I'|, k'$ are bounded by $g(k)$ for some $g$ and (ii) given an $\beta$-approximate solution $s'$ for $(I', k')$, $\cB$ produces a solution $s$ of $(I, k)$ such that $s$ is an $(\alpha\beta)$-approximate solution\footnote{We use a similar convention here as our convention for approximation ratios. That is, $\alpha \leqs 1$ for maximization problems and $\alpha \geqs 1$ for minimization problems. Note that this is not the same as in~\cite{LPRS17} where $\alpha \geqs 1$ in both cases; nevertheless, it is simple to see that these different conventions do not effect any of the results.} for $(I, k)$. Akin to (exact) kernelization, Lokshtanov \etal~\cite{LPRS17} shows that a decidable parameterized optimization problem admits $\alpha$-approximate kernel if and only if it admits an FPT $\alpha$-approximation algorithm. (We refer interested readers to Section 2.1 of~\cite{LPRS17} for more details.) In light of Marx's algorithm for Max $k$-VC~\cite{Marx08}, this immediately implies that Max $k$-VC admits $(1 - \varepsilon)$-approximate kernel for any $\varepsilon > 0$. Lokshtanov \etal~\cite{LPRS17} made this bound more specific, by showing that the insights from Marx's work can be turned into an $(1 - \varepsilon)$-approximate kernel where the number of vertices in the new instance is at most $O(k^5/\varepsilon^2)$.

{\bf Minimum $k$-Vertex Cover.} We will also consider the minimization variant of the Min $k$-VC, which we call \emph{Minimum $k$-Vertex Cover (Min $k$-VC)}. The goal of this problem is to find a subset of $k$ vertices that \emph{minimizes} the number of edges covered. Note that this is \emph{not} a natural relaxation of Vertex Cover and is in fact more closely related to edge expansion problems. (See~\cite{GandhiK15} and discussion therein for more information.) The greedy algorithm that picks $k$ vertices with minimum degrees yields a 2-approximation. Gandhi and Kortsarz~\cite{GandhiK15} showed that this is likely tight: assuming the Small Set Expansion Conjecture~\cite{RS10}, it is hard to approximate Min $k$-VC to within $(2 - \varepsilon)$ factor for any $\varepsilon > 0$. As for its parameterized complexity, similar to Max $k$-VC, Min $k$-VC is \W[1]-hard~\cite{GuoNW07} and admits an FPT-AS with running time $(k/\varepsilon)^{O(k)}$~\cite{GLL18a,GLL18}.

{\bf Weight vs Unweighted.} All results stated above are for unweighted graphs. The natural extensions of Max $k$-VC (resp. Min $k$-VC) to edge-weighted graphs ask to find subsets of vertices of size $k$ that maximizes (resp. minimizes) the total weight of the edges covered. To avoid confusion, we refer to these weighted variants explicitly as Weighed Max $k$-VC and Weight Min $k$-VC. Clearly, since these problems are more general than the unweighted ones, the lower bounds above (including inapproximability results and \W[1]-hardness) applies immediately. It is also quite simple to check that all aforementioned polynomial time approximation algorithms for the unweighted case extends naturally to the weighted setting too. The FPT-ASes are slightly trickier, but Gupta \etal~\cite{GLL18a} provide an argument discretizing the weights and extend their FPT-ASes to the weighted case with similar time complexity. It is also possible to apply this argument to Lokshtanov \etal's~\cite{LPRS17} approximate kernel, although it would result in a graph of $O(k^7/\varepsilon^4)$ vertices instead of $O(k^5/\varepsilon^2)$ for the unweighted case.

\subsection{Our Results} \label{sec:results}

For convenience, all our results stated below are for the weighted version of the problems, and moreover we allow self-loops in the input graph. This is the most general version of the problem and, hence, the algorithmic results below apply directly to the unweighted case (nd the weighted simple graph case. We also note that this choice is partly motivated by the fact that in some applications, such Gupta \etal's~\cite{GLL18a,GLL18} algorithms for Minimum $k$-Cut, this full generality is needed. (Unfortunately, our result does not imply faster algorithms for Minimum $k$-Cut, as the bottlenecks of Gupta \etal's approach is elsewhere\footnote{For their FPT approximation algorithm~\cite{GLL18a}, the bottleneck is in the reduction from Min $k$-Cut to Laminar $k$-Cut which runs in time $2^{O(k^2)} \cdot \poly(n)$. For their $(1 + \varepsilon)$-approximation algorithm~\cite{GLL18}, the bottleneck is in the dynamic programming step which takes $(k/\varepsilon)^{O(k)}\cdot\poly(n)$ time.}.)

We remark that, while the algorithmic results apply directly to the more restricted version, the approximate kernel does \emph{not}. This is because, in a more restricted version (e.g. unweighted) of the problems, the instance output by the reduction algorithm is also more restrictive (e.g. unweighted), meaning that one cannot simply use the approximate kernel for the more general version. Nevertheless, as we will point out below, our approximate kernel also extends to the unweighted setting (and simple graph setting), with a small loss in parameter. 

\subsection*{Maximum $k$-Vertex Cover}

Our first result is a faster FPT-AS for Max $k$-VC that runs in time $O(1/\varepsilon)^k \cdot \poly(n)$, which improves upon a $(k/\varepsilon)^{k} \cdot \poly(n)$-time FPT-AS due to Gupta, Lee and Li~\cite{GLL18}.

\begin{theorem} \label{thm:fpt-as}
For every $\varepsilon > 0$, there exists an $(1 - \varepsilon)$-approximation algorithm for Weighted Max $k$-VC that runs in time $O(1/\varepsilon)^k \cdot \poly(n)$.
\end{theorem}

Perhaps more importantly, our FPT-AS is simple and yields a new insight compared to the previous FPT-ASes~\cite{Marx08,GLL18a,GLL18}. In particular, our algorithm is just the following: restrict ourselves only to the $O(k/\varepsilon)$ vertices with maximum weighted degrees and use brute force to find a $k$-vertex subset among these vertices that cover edges with maximum total weight.

To demonstrate the differences to the previous algorithms, let us briefly sketch how they work here. The known FPT-ASes~\cite{Marx08,GLL18a,GLL18} all rely on a degree-based argument for the unweighted case due to Marx~\cite{Marx08} who consider the following two cases:
\begin{enumerate}
\item The vertex with maximum degree have degree at least $k^2/\varepsilon$. In this case, one can simply take the $k$ vertices with largest degree because the number of edges with both endpoints in the set is at most $\binom{k}{2}$, meaning that it only affects the number of edges covered by at most an $\varepsilon$ factor and thus this is already an $(1 - \varepsilon)$-approximation for the problem.
\item The vertex with maximum degree have degree at most $k^2/\varepsilon$. The key property in this case is that the number of edges covered by the optimal solution is at most $k^3/\varepsilon$, which is bounded by a function of $k$. Marx's algorithm then proceeds as follows: (i) guess the number of edges $\ell \leqs k^3/\varepsilon$ in the optimal solution, (ii) guess (among the $k^\ell$ possibilities) which vertex (in the solution) that each edge is covered by, (iii) randomly color each edge in the input graph with one of $\ell$ colors and randomly color each vertex with one of $k$ colors and (iv) finally, determine whether there are $k$ vertices each of different color that covers edges with colors as guessed in Step (ii). Note that Step (iv) can be easily done in polynomial time. Since $\ell$ is bounded by $k^3/\varepsilon$, the algorithm succeeds with probability at least $k^{-O(k^3/\varepsilon)}$, which can be turned into a randomized algorithm with running time $k^{-O(k^3/\varepsilon)} \cdot \poly(n)$ that succeeds with high probability. Finally, it can be derandomized using standard techniques (see~\cite{AYZ95}). 
\end{enumerate}
The speed-up of Gupta \etal~\cite{GLL18a,GLL18} comes from the change in the second case. Roughly speaking, they show that more elaborated coloring techniques can be used, in conjunction with dynamic programming, to speed the second case up to $(k/\varepsilon)^{O(k)} \cdot \poly(n)$.

Intuitively, our result shows that this case-based analysis is in fact not needed, as it suffices to consider the $O(k/\varepsilon)$ vertices with highest weighted degrees. Moreover, a nice feature about our algorithm is that it works naturally for the weighted case, whereas Gupta \etal\ needs to employ a discretization argument to deal with this case. (See Section 5.2 in the full version of~\cite{GLL18a}.)

Another feature of our algorithm is that it immediately gives an approximate kernelization for the problem, by restricting to the subgraph induced by the $O(k/\varepsilon)$ vertices and adding self-loops with appropriate weights to compensate the edges from these vertices to the remaining vertices. This results in an $(1 - \varepsilon)$-approximate kernelization of $O(k/\varepsilon)$ vertices for Max $k$-VC: 

\begin{lemma} \label{lem:main}
For every $\varepsilon > 0$, Weighted Max $k$-VC admits an $(1 - \varepsilon)$-approximate kernelization with $O(k/\varepsilon)$ vertices.
\end{lemma}

As stated earlier, the above result is not directly comparable to Lokshtanov \etal's~\cite{LPRS17} approximate kernel of $O(k^5/\varepsilon^2)$ vertices for the unweighted version of Max $k$-VC. Fortunately, our technique also gives an $O(k/\varepsilon^2)$-vertex approximate kernel for the unweighted case, which indeed improves upon Lokshtanov \etal's result. (See the end of Section~\ref{sec:kernel}.)

Interestingly, the above approximate kernelization also has an application outside of parameterized complexity: using our approximate kernelization as a preprocessing step, we can directly apply Raghavendra and Tan's SDP-based algorithm for 2SAT with cardinality constraint~\cite{RT12} to give an $0.92$-approximation algorithm for Max $k$-VC in polynomial time. This improves upon the aforementioned polynomial time approximation algorithm of Feige and Langberg~\cite{FL01} which yields $(0.75 + \delta)$-approximation for some (small and unspecified) constant $\delta > 0$.

\begin{corollary} \label{cor:approx}
There exists a polynomial time 0.92-approximation algorithm for Weighted Max $k$-VC.
\end{corollary}

We note here that the approximation guarantee above is even better than the previous best known ratios for some special cases, such as in bipartite graph~\cite{ApollonioS14,BonnetEPS18} where the previous best known approximation ratio is 0.821~\cite{BonnetEPS18}.

\subsection*{Minimum $k$-Vertex Cover}

For the Weighted Min $k$-VC problem, we give a FPT-AS with similar running time of $O(1/\varepsilon)^{O(k)} \cdot \poly(n)$ for the problem. Once again, this improves upon the $(k/\varepsilon)^{O(k)} \cdot \poly(n)$-time algorithm of Gupta \etal~\cite{GLL18a,GLL18}.

\begin{theorem} \label{thm:minkvc}
For every $\varepsilon > 0$, there exists an $(1 + \varepsilon)$-approximation algorithm for Weighted Min $k$-VC that runs in time $O(1/\varepsilon)^{O(k)} \cdot \poly(n)$.
\end{theorem}

We remark that this algorithm is different from the algorithm for Max $k$-VC and is instead based on a careful branch-and-bound approach. A natural question here is perhaps whether this difference is inherent. While it is unclear how to make this question precise, we provide an evidence that the two problems are indeed of different natures by showing that, in contrast to Max $k$-VC, a polynomial size approximate kernelization for Min $k$-VC for any factor less than two is unlikely to exist:

\begin{lemma} \label{lem:no-kernel}
Assuming the Strong Small Set Expansion Hypothesis (Conjecture~\ref{conj:strong-sseh}) and $\NP \nsubseteq \coNP/\poly$, Weighted Min $k$-VC does not admit a polynomial size $(2 - \varepsilon)$-approximate kernelization for any $\varepsilon \in (0, 1]$.
\end{lemma}

The above result is under a variant of the Small Set Expansion Hypothesis~\cite{RS10}; please refer to Section~\ref{sec:no-kernel} for the precise definition of the variant. We also note that the above lower bound also applies to the unweighted version; again please see Section~\ref{sec:no-kernel} for more details.

\section{Notations}

Throughout this note, we think of an edge-weighted graph as a complete graph (self-loops included) where each edge is endowed with a non-negative weight. More specifically, an edge-weighted graph $G$ consists of a vertex set $V_G$ and a weight function $w_G: \binom{V_G}{\leqs 2} \to \mathbb{R}_{\geqs 0}$. (Note that, for a set $U$ and a non-negative integer $\ell$, we use $\binom{U}{\leqs \ell}$ and $\binom{U}{\ell}$ to denote the collections of subsets of $U$ of sizes at most $\ell$ and exactly $\ell$ respectively.) When the graph is clear from the context, we may drop the subscript $G$, and we sometimes use $w_e$ to denote $w(e)$ for brevity. For each vertex $v \in V$, we use $\wdeg(v)$ to denote its weighted degree, i.e., $\wdeg(v) = \sum_{e \in \binom{V_G}{\leqs 2}, v \in e} w_e$. For a subset $S \subseteq V_G$, we write $\wdeg(S)$ to denote $\sum_{v \in S} \wdeg(v)$. For subsets $S, T \subseteq V_G$, we use $E_G(S, T)$ to denote the total weight of edges with at least one endpoint in $S$ and at least one endpoint in $T$; more specifically, $E_G(S, T) = \sum_{e \in \binom{V_G}{\leqs 2}, e \cap S \ne \emptyset, e \cap T \ne \emptyset} w_G(e)$. Note that $E_G(S, S)$ is the total weight of the edges covered by $S$; for brevity, we use $E_G(S)$ as a shorthand for $E_G(S, S)$. Finally, we use $\optmin(G, k)$ and $\optmax(G, k)$ to denote the optimums of Min $k$-VC and Max $k$-VC respectively on the instance $(G, k)$. More formally, $\optmin(G, k) = \min_{S \in \binom{V_G}{k}} E_G(S)$ and $\optmax(G, k) = \max_{S \in \binom{V_G}{k}} E_G(S)$.

\section{Maximum $k$-Vertex Cover}

We will now prove our results for Max $k$-VC. To do so, it will be convenient to order the vertices of the input graph $V_G$ based on their weighted degree (ties broken arbitrarily), i.e., let $v_1, \dots, v_n$ be the ordering of vertices in $V_G$ such that $\wdeg_G(v_1) \geqs \cdots \geqs \cdot \wdeg_G(v_n)$. Moreover, we use $V_i$ to denote the set of $i$ vertices with highest weighted degree, i.e., $V_i = \{v_1, \dots, v_i\}$.

\subsection{A Simple Observation and A Faster FPT-AS}

Our main insight to the Weighted Max $k$-VC problem is that there is always an $(1 - \varepsilon)$-approximate solution which is entirely contained in $V_{O(k/\varepsilon)}$, as stated more formally below.

\begin{observation} \label{obs:largest-deg}
For any $\varepsilon > 0$, let $n' = \min\{k + \lceil k/\varepsilon \rceil, n\}$. Then, there exists $S^* \subseteq V_{n'}$ of size $k$ such that $E_G(S^*) \geqs (1 - \varepsilon) \cdot \optmax(G, k)$.
\end{observation}

Note that this implies Theorem~\ref{thm:fpt-as}: we can enumerate all $k$-vertex subsets of $V_{n'}$ and find an $(1 - \varepsilon)$-approximation for Max $k$-VC in $\binom{|V_{n'}|}{k} \poly(n) = O(n'/k)^k\poly(n) = O(1/\varepsilon)^k\poly(n)$ time.

Before we present a formal proof of the observation, let us briefly give an (informal) intuition behind the proof. Let $S_{\opt}$ be the optimal solution for $(G, k)$. Our goal is to construct another set $S^* \subseteq V_{n'}$ such that $E_G(S^*)$ is roughly the same as $E_G(S_{\opt})$. To do so, we will just replace each vertex in $S_{\opt} \setminus V_{n'}$ by a vertex in $V_{n'} \setminus S_{\opt}$. Intuitively, this should be good for the solution, as we are replacing one vertex with another vertex that has higher weighted degree. However, this argument does not yet work: we might ``double count'' edges with both endpoints coming from the new vertices. The key point here is that, while we will not be able to avoid this double counting completely, we will be able to pick new vertices such that the total weight of such doubled counted edges is small. This is just because the set $V_{n'}$ is so large that even if we pick a random $k$ vertices from it, the probability that a given added edge is double counted is only $O(\varepsilon)$.


\begin{proof}[Proof of Observation~\ref{obs:largest-deg}]
Note that, if $n' = n$, the statement is obviously true. Hence, we may assume that $n' = k + \lceil k/\varepsilon\rceil$. Let $S_{\opt} \subseteq V_G$ denote any optimal solution, i.e., any subset of $V_G$ of size $k$ with $E_G(S_{\opt}) = \optmax(G, k)$. Let $S_{\opt}^{\text{in}} = S_{\opt} \cap V_{n'}$, $S_{\opt}^{\text{out}} = S_{\opt} \setminus V_{n'}$ and $U = V_{n'} \setminus S_{\opt}$. 

We construct $S \subseteq V_{n'}$ in randomly as follows. We randomly select a subset $U^* \subseteq U$ of $|S^{\text{out}}_{\opt}|$ vertices uniformly at random, and let $S = S_{\opt}^{\text{in}} \cup U^*$. Clearly, $S$ is a subset of $V_{n'}$ of size $k$. We will show that the expected value of $E_G(S)$ is at least $(1 - \varepsilon) \cdot \optmax(G, k)$. This would imply that there exists $S^* \subseteq V_{n'}$ of size $k$ such that $E_G(S^*) \geqs (1 - \varepsilon) \cdot \optmax(G, k)$ as desired.

To bound $\Ex[E_G(S)]$, let us first rearrange $E_G(S)$ as follows.
\begin{align} \label{eq:identity}
E_G(S) &= E_G(S_{\opt}^{\text{in}}) + E_G(U^*) - E_G(U^*, S_{\opt}^{\text{in}}).
\end{align}
Let $\rho = |S_{\opt}^{\text{out}}|/|U|$; note here that $\rho \leqs k/(n' - k) \leqs \varepsilon$. We can now bound $\Ex[E_G(U^*, S_{\opt}^{\text{in}})]$ by
\begin{align}
\Ex[E_G(U^*, S_{\opt}^{\text{in}})] = \sum_{u \in U} \sum_{v \in S_{\opt}^{\text{in}}} w_{\{u, v\}} \cdot \Pr[u \in U^*] = \rho \cdot \sum_{u \in U} \sum_{v \in S_{\opt}^{\text{in}}} w_{\{u, v\}} \leqs \varepsilon \cdot E_G(S_{\opt}^{\text{in}}) \label{eq:tmp01}
\end{align}
Moreover, $\Ex[E_G(U^*)]$ can be rearranged as
\begin{align}
\Ex[E_G(U^*)] 
&= \Ex\left[\sum_{u \in U^*} \left(\wdeg(u) - \frac{1}{2} \sum_{v \in U^* \setminus \{u\}} w_{\{u, v\}}\right)\right] \nonumber \\
&= \Ex\left[\sum_{u \in U} \left(\wdeg(u) \cdot \ind[u \in U^*] - \frac{1}{2} \sum_{v \in U \setminus \{u\}} w_{\{u, v\}} \cdot \ind[u \in U^* \wedge v \in U^*] \right)\right] \nonumber\\
&= \sum_{u \in U} \left(\wdeg(u) \cdot \Pr[u \in U^*] - \frac{1}{2} \sum_{v \in U \setminus \{u\}} w_{\{u, v\}} \cdot \Pr[u \in U^* \wedge v \in U^*] \right) \nonumber\\
&\geqs \sum_{u \in U} \left(\wdeg(u) \cdot \rho - \frac{1}{2} \sum_{v \in U \setminus \{u\}} w_{\{u, v\}} \cdot \rho^2 \right) \nonumber\\
&\geqs \rho(1 - \rho/2) \cdot \left(\sum_{u \in U} \wdeg(u)\right) \nonumber \\
&\geqs \rho(1 - \varepsilon) \cdot \left(\sum_{u \in U} \wdeg(u)\right) \label{eq:tmp1}
\end{align}
where in the first inequality we use the fact that $\Pr[u \in U^* \wedge v \in U^*] \leqs \Pr[u \in U^*]\Pr[v \in U^*] = \rho^2$.

Recall that the vertices are sorted in decreasing order of degrees; thus, for all $u \in U$, we have $\wdeg(u) \geqs \left(\sum_{v \in S_{\opt}^{\text{out}}} \wdeg(v)\right)/|S_{\opt}^{\text{out}}| \geqs E_G(S_{\opt}^{\text{out}})/|S_{\opt}^{\text{out}}|$. From this and~\eqref{eq:tmp1}, we arrive at
\begin{align}
\Ex[E_G(U^*)] &\geqs \rho(1 - \varepsilon) \cdot |U| \cdot \left(E_G(S_{\opt}^{\text{out}})/|S_{\opt}^{\text{out}}|\right) = (1 - \varepsilon) \cdot E_G(S_{\opt}^{\text{out}}) \label{eq:tmp2}
\end{align}

Plugging~\eqref{eq:tmp01} and~\eqref{eq:tmp2} back into~\eqref{eq:identity}, we indeed have
\begin{align*}
\Ex[E_G(S)] \geqs (1 - \varepsilon)(E_G(S_{\opt}^{\text{in}}) + E_G(S_{\opt}^{\text{out}})) \geqs (1 - \varepsilon) \cdot E_G(S_{\opt}) = (1 - \varepsilon) \cdot \optmax(G, k),
\end{align*}
which concludes the proof.
\end{proof}

\subsection{An Approximate Kernel} \label{sec:kernel}

Observation~\ref{obs:largest-deg} also naturally gives an $(1 - \varepsilon)$-approximate kernel for Weighted Max $k$-VC where the new instance has $O(k/\varepsilon)$ vertices, as stated below.

\begin{proof}[Proof of Lemma~\ref{lem:main}]
The reduction algorithm $\cA$ works by taking the graph induced on $V_{n'}$ (where $n' = \min\{k + \lceil k/\varepsilon \rceil, n\}$ as in Observation~\ref{obs:largest-deg}) and add appropriate weights to self-loops to compensate for edges going out of $V_{n'}$. More precisely, $\cA$ outputs $(G', k)$ where $V_{G'} = V_{n'}$ and $w_G'(\{u, v\}) = w_G'(\{u, v\})$ for all $u \ne v \in V_{G'}$ and $w_{G'}(u) = w_G(u) + E_G(\{u\}, V_G \setminus V_{n'})$ for all $u \in V_{G'}$.

The solution lifting algorithm $\cB$ simply outputs the same solution as its get. It is obvious to see that $E_{G'}(S) = E_G(S)$. Hence, if $E_{G'}(S) \geqs \alpha \cdot \optmax(G', k)$, then Observation~\ref{obs:largest-deg} implies that $E_G(S) = E_{G'}(S) \geqs \alpha(1 - \varepsilon) \cdot \optmax(G, k)$. This means that $(\cA, \cB)$ is an $(1 - \varepsilon)$-approximate kernel; moreover, it is obvious that the graph output by $\cA$ has size $O(k/\varepsilon)$ as desired.
\end{proof}

As mentioned earlier, the above kernel does not directly work for the unweighted case. Let us sketch below how we can modify the above proof to work in this case, albeit with a slightly worse $O(k/\varepsilon^2)$ vertices in the reduced instance. We omit the full proof, which is a simple undergraduate-level exercise, and only describe the main ideas. We do this in two steps; we first modify the proof for weighted graphs without self-loops and then we proceed to unweighted graphs.
\begin{itemize}
\item Suppose that the graphs $G$ and $G'$ must not contain any self-loops. Then, instead of adding self-loops as above, $\cA$ will add $n_{\text{padded}} = \lceil kn'/\varepsilon \rceil = O(k/\varepsilon^2)$ padded vertices and let the weight between each padded vertex and $u \in V_{n'}$ be $\frac{E_G(\{u\}, V_G \setminus V_{n'})}{n_{\text{padded}}}$. Once again, if we take a look at any set $S \subseteq V_{n'}$, we immediately have $E_G(S) = E_{G'}(S)$. The only additional argument needed is that these padded vertices has little effect on any solution. Indeed, it is simple to see that the weighted degree of each padded vertex is at most $(\varepsilon/k) \cdot \optmax(G, k)$. Thus, throwing these vertices away from any subset of size $k$ affect the total weights of edges covered by at most $\varepsilon \cdot \optmax(G, k)$, which implies that this is an $(1 - 2\varepsilon)$-approximate kernel. Adjusting $\varepsilon$ appropriately gives the $(1 - \varepsilon)$-approximate kernel with $O(k/\varepsilon^2)$ vertices.
\item The above idea naturally adapts to the unweighted case. Instead of adding an edge from every $u \in V_{n'}$ to all the padded vertices, we just add $E_G(\{u\}, V_G \setminus V_{n'})$ edges from each $u \in V_{n'}$ to different padded vertices. These edges are added in a way that each padded vertices has roughly the same degree. It is simple to check that, if the degree of all vertices $u \in V_{n'}$ is at most say $k/\varepsilon^2$, then this works immediately (with the same proof as above). The only issue is when there are vertices with degree larger than $k/\varepsilon^2$. (In this case, the number of edges required to be added may even be larger than $n_{\text{padded}}$!) Nevertheless, this issue can also be easily resolved, by observing that, if any vertex in $V_{k}$ has degree at least $k/\varepsilon$, then we can always take it in our solution while guaranteeing that the solution still remains within $\varepsilon \cdot \optmax(G, k)$ of the optimum. Hence, the reduction algorithm can first greedily pick these vertices and then use the padded argument as above; since no large degree vertex remains, the proof of the second step now works and we have the desired approximate kernel.
\end{itemize}

\subsection{Raghavendra-Tan Algorithm and An Improved Approximation}

We next describe how our approximate kernel can be used a preprocessing step for the aforementioned algorithm of Raghavendra and Tan~\cite{RT12} for Max 2SAT with cardinality constraint to obtain improved approximation for Weighted Max $k$-VC. 

Recall that the (weighted) Max 2SAT with cardinality constraint is the following problem. Given a collection $\cC$ of conjunctions of at most two literals (of variables $\{x_1, \dots, x_n\}$) and their associated weights, find an assignment to $\{x_1, \dots, x_n\}$ satisfying $x_1 + \cdots + x_n = k$ that maximizes the total weights of satisfied clauses in $\cC$. Raghavendra and Tan~\cite{RT12} device an algorithm with approximation ratio strictly greater than 0.92 for the problem, as stated below.

\begin{theorem}[\cite{RT12}] \label{thm:2sat}
For some $\alpha > 0.92$, there exists an $\alpha$-approximation algorithm for Max 2SAT with cardinality constraint that runs in time\footnote{The running time of the algorithm is not stated in this form in~\cite{RT12} as they are only concerned about the case where $k = \Omega(n)$, for which the running time is polynomial. To see that the running time is of the form $n^{\poly(n/k)}$, we note that their algorithm needs the variance guaranteed in their Theorem 5.1 to be at most $\poly(k/n)$. This means that they need the SDP solution to be $\poly(k/n)$-independence; to find such a solution, the running time required is $n^{\poly(n/k)}$ (see Theorem 4.1 in that paper).} $n^{\poly(n/k)}$.
\end{theorem}

It is not hard to see that the Weighted Max $k$-VC can be formulated as Max 2SAT with cardinality constraint: we create a variable $x_i$ for each vertex $v_i$, and, for each $\{v_i, v_j\} \in \binom{V_G}{\leqs 2}$, we create a clause $(v_i \vee v_j)$ with weight $w_{\{v_i, v_j\}}$. Obviously, any solution to Max 2SAT satisfying $x_1 + \cdots + x_n = k$ is also a solution of Max $k$-VC with the same cost. Of course, the only issue in applying Raghavendra and Tan's algorithm here is that its running time $n^{\poly(n/k)}$ is not polynomial when $k = o(n)$. Fortunately, our approximate kernel above precisely circumvents this issue, as the reduction algorithm produces an instance $(G', k)$ where $|V_{G'}| \leqs O(k/\varepsilon)$. Thus, we can now apply the algorithm and arrives at 0.92 approximation for Weight Max $k$-VC in polynomial time.

\begin{proof}[Proof of Corollary~\ref{cor:approx}]
Let $\alpha$ be the approximation ratio from Theorem~\ref{thm:2sat} and let $\varepsilon > 0$ be a sufficiently small constant such that $\alpha(1 - \varepsilon) \geqs 0.92$. Let $\cA$ be the reduction algorithm for the $(1 - \varepsilon)$-approximate kernel as defined in the proof of Lemma~\ref{lem:main}.

For any instance $(G, k)$ of Weight Max $k$-VC, we apply $\cA$ to arrive at a reduced instance $(G', k)$ where $|V_{G'}| \leqs O(k/\varepsilon)$. We then formulate the instance $(G', k)$ as an instance of Max 2SAT with cardinality constraint and apply the Raghavendra-Tan algorithm, which gives an $\alpha$-approximate solution, i.e., a set $S \subseteq V_{G'}$ of size $k$ such that $E_{G'}(S) \geqs \alpha \cdot \optmax(G', k) \geqs \alpha (1 - \varepsilon) \cdot \optmax(G, k) \geqs 0.92 \cdot \optmax(G, k)$. Note that the Raghavendra-Tan algorithm runs in $k^{\poly(|V_{G'}|/k)} = k^{\poly(1/\varepsilon)}$ time. Hence, we have found a 0.92-approximate solution for $(G, k)$ in polynomial time.
\end{proof}


\section{Minimum $k$-Vertex Cover}

\subsection{A Faster FPT-AS}

We now present our result on Weighted Min $k$-VC, starting with the faster FPT-AS (Theorem~\ref{thm:minkvc}). It will be more convenient for us to work with a multicolored version of the problem, which we call \emph{Multicolored Min $k$-VC}. In Multicolored Min $k$-VC, we are given $G, k$ as before and also a coloring $\chi: V_G \to [k]$. A set $S \subseteq V_G$ is said to be \emph{colorful} if every vertex in $S$ is assigned a different color, i.e., $|\chi(S)| = |S|$. The goal of Multicolored Min $k$-VC is to find a colorful $S \subseteq V_G$ of size $k$ that maximizes $E_G(S)$. We overload the notation $\optmin$ and also use it to denote the optimum of Multicolored Min $k$-VC; that is, we let $\optmin(G, k, \chi) = \min_{S \in \binom{V_G}{k}, |\chi(S)| = k} E_G(S)$.

The main theorem of this section is the following FPT-AS for Multicolored Min $k$-VC. 

\begin{theorem} \label{thm:minkvc-colorful}
For any $\varepsilon > 0$, there exists an $(1 + \varepsilon)$-approximation algorithm for Multicolored Min $k$-VC that runs in time $O(1/\varepsilon)^{O(k)} \cdot \poly(n)$.
\end{theorem}

We note here that the above lemma immediately gives an FPT-AS for (uncolored) Weight Min $k$-VC with similar running time (i.e. Theorem~\ref{thm:minkvc}) via standard color-coding technique~\cite{AYZ95}. Specifically, they show how to construct a family $\cF$ of $k$-perfect hash functions from $V_G \to \{1, \dots, k\}$ in $2^{O(k)} \cdot \poly(n)$ time. By running the FPT-AS from Theorem~\ref{thm:minkvc-colorful} on $(G, k, \chi)$ for all $\chi \in \cF$ and take the best solution among the outputs, we arrive at the FPT-AS for (uncolored) Weight Min $k$-VC.

We now proceed to discuss the intuition behind Theorem~\ref{thm:minkvc-colorful}. The algorithm consists of two parts: subgraph generation and dynamic programming. Roughly speaking, the subgraph generation part will, for each set of colors $C \subseteq [k]$, generate connected colorful subsets $T \subseteq V_G$ whose color is $C$ and record the minimum $E_G(T)$ in the table cell \text{DP}$[C]$. The second part of the algorithm then uses a standard dynamic programming to find a colorful $k$-vertex $S$ with minimum $E_G(S)$.

For the purpose of exposition, let us assume for the moment that our graph is unweighted. The subgraph generation part is the heart of the algorithm, and, if not implemented in a careful manner, will be too slow. For instance, the trivial implementation of this is as a recursive function that maintains a set of included vertices $\Sinc$ and a set of active vertices $\Sact$. This function then picks any vertex $u \in \Sact$ and tries to select at most $k$ neighbors of $u$ to add into $\Sinc$ and $\Sact$; the function then remove $u$ from $\Sact$ and recursively call itself on this new sets. (Note that in this step it also makes sure that the set $\Sinc$ remains colorful; otherwise, the recursive call is not made.) The function stops when $\Sact$ is empty and update \text{DP}$[C]$ to be the minimum between the current value and $E_G(\Sinc)$. As the reader may have already noticed, while this algorithm records (exactly) the correct answer into the table, it is very slow. In particular, if say we run this on a complete graph, then it will generates $n^{\Theta(k)}$ subgraphs.

The algorithm of Gupta, Lee and Li~\cite{GLL18a,GLL18}, while not stated in this exact form, can be viewed as a more careful implementation of this approach. In particular, they use the observation of Marx~\cite{Marx08} (that was also outlined outline in Section~\ref{sec:results}): for unweighted graphs, if the optimal solution has any vertex with degree at least $\binom{k}{2}/\varepsilon$, simply picking the $k$ vertices with minimum degrees would already be an $(1 + \varepsilon)$-approximate solution. In other words, one may assume that the graph has degree bounded by $\binom{k}{2}/\varepsilon = O(k^2/\varepsilon)$. When this is the case, the algorithm from the previous paragraph in fact runs in $O(k/\varepsilon)^{O(k)} \cdot \poly(n)$ time; the reason is that the number of choices to be made when adding a vertex is only $O(k^2/\varepsilon)$ instead of $n$ as before. Hence, the running time becomes $O(k^2/\varepsilon)^k \cdot \poly(n) = (k/\varepsilon)^{O(k)} \cdot \poly(n)$.

To obtain further speed up, we observe that, if at most $\varepsilon/2$ fraction of neighbors of a vertex $u$ lies in the optimal solution, then ignoring all of them completely while branching would change the number of covered edges by factor of no more than $\varepsilon$. (This is shown formally in the proof below.) In other words, instead of trying all subsets of at most $k$ neighbors of $u$. We may only try subsets with at least $d \varepsilon / 2$ (and at most $k$) neighbors of $u$ where $d$ is the degree of $u$. The point here is that, while there are still $\exp(d)$ branches, we are adding at least $d\varepsilon/2$ vertices. Hence, the ``branching factor per vertex added'' is small: namely, for $j \geqs d\varepsilon/2$, the ``branching factor per vertex added'' is only $\binom{d}{j}^{1/j} \leqs ed/j \leqs O(1/\varepsilon)$. This indeed gives the running time of $O(1/\varepsilon)^{O(k)} \cdot \poly(n)$. (Note that such branching may result in a connected component being separated; however, when this is the case, the number of edges between the generated parts must be small anyway.)

Let us now shift our discussion to the edge-weighted graph case. Once again, as we will show formally in the proof, throwing away the edges adjacent to $u$ with total weight at most $(\varepsilon/2) \cdot \wdeg(u)$ only affects the solution value by no more than $\varepsilon$ factor. However, this observation alone is not enough; specifically, unlike the unweighted case, this does not guarantee that many vertices must be selected. As an example, if there is a vertex $v$ where $w_{\{u, v\}} = 0.5 \cdot \wdeg(u)$, then even the set $\{v\}$ should be consider when we branch. Nevertheless, it is once again possible to show that, we can select a collection $\cT$ of representative subsets such that, for any set $S \subseteq V_G$ (the true optimal set), we can arrive in a subset in $\cT$ by throwing away vertices whose edges to $u$ are of total weight at most $(\varepsilon/2) \cdot \wdeg(u)$. In other words, it is ``safe'' to just consider branching with subsets in $\cT$ instead of all subsets. Again, the collection $\cT$ will satisfy the property that the ``branching factor per vertex added'' is small; that is, for any $j$, the number of $j$-element subsets that belong to $\cT$ is at most $O(1/\varepsilon)^j$. The existence and efficient construction of such $\cT$ is stated below in a more general form. Note that, in the context of subgraph generation algorithm, one should think of $\delta = \varepsilon / 2$, $\ell = n - 1$ (all vertices except $u$ itself) and $P = \wdeg(u) - w_{\{u\}}$.

\begin{lemma} \label{lem:eps-net}
Let $a_1, \dots, a_\ell \geqs 0$ be any non-negative real numbers, let $\delta > 0$ be any positive real number, and let $P = \sum_{i \in [\ell]} a_i$. Then, there exists a collection $\cT$ of subsets of $[\ell]$ such that
\begin{enumerate}[(i)]
\item For all $j \in [\ell]$, we have $\left|\cT \cap \binom{[\ell]}{j}\right| \leqs O(1/\delta)^j$, and, \label{property:size-bound}
\item For any $S \subseteq [\ell]$, there exists $T \in \cT$ such that $T \subseteq S$ and $\sum_{i \in (S \setminus T)} a_i \leqs \delta \cdot P$. \label{property:eps-net}
\end{enumerate} 
Moreover, for any $j \in [\ell]$, $\cT \cap \binom{[\ell]}{\leqs j}$ can be computed in $O(1/\delta)^{O(j)}\ell^{O(1)}$ time.
\end{lemma}

\begin{proof}
Let $\pi: [\ell] \to [\ell]$ be any permutation such that $a_{\pi(1)} \geqs \cdots \geqs a_{\pi(\ell)}$. For each $j \in [\ell]$, we construct $\cT \cap \binom{[\ell]}{j}$ by taking all $j$-element subsets of $\{\pi(1), \dots, \pi(\min\{j \cdot \lceil 1/\delta\rceil, \ell\})\}$. We have
\begin{align*}
\left|\cT \cap \binom{[\ell]}{j}\right| \leqs \binom{j \cdot \lceil 1/\delta\rceil}{j} \leqs \left(\frac{e j \cdot \lceil 1/\delta\rceil}{j}\right)^j \leqs O(1/\delta)^j.
\end{align*}
Moreover, it is clear that the set $\cT \cap \binom{[\ell]}{j}$ can be generated in time polynomial in the size of the set and $\ell$, which is $O(1/\delta)^{O(j)}\ell^{O(1)}$ as desired.

Finally, we will prove~\ref{property:eps-net}. Consider any subset $S \subseteq [\ell]$ and suppose that its elements are $\pi(i_1), \dots, \pi(i_m)$. We pick the set $T$ as follows: let $t$ be the largest index such that $i_t \leqs t \cdot \lceil 1/\delta \rceil$ and let $T = \{\pi(i_1), \dots, \pi(i_t)\}$. Since $i_t \leqs t \cdot \lceil 1/\delta \rceil$, $T$ is a $t$-element subset from $\{\pi(1), \dots, \pi(\min\{t \cdot \lceil 1/\delta \rceil, \ell\})$ and hence $T$ belongs to $\cT$. To prove~\ref{property:eps-net}, observe that, by definition of $t$, we have $i_g > g \cdot \lceil 1/\delta \rceil$ for all $g > t$. This means that
\begin{align*}
\sum_{i \in (S \setminus T)} a_i = \sum_{g=t+1}^{m} a_{\pi(g)} \leqs \sum_{g=t+1}^{m}\left(\frac{1}{\lceil 1/\delta \rceil} \sum_{i=(g - 1) \cdot \lceil 1/\delta \rceil + 1}^{g \cdot \lceil 1/\delta \rceil} a_i\right) &\leqs \frac{1}{\lceil 1/\delta \rceil} \sum_{i\in[\ell]} a_i \leqs \delta \cdot P,
\end{align*}
which concludes the proof.
\end{proof}

With the above lemma ready, we proceed to the proof of Theorem~\ref{thm:minkvc-colorful}.

\begin{proof}[Proof of Theorem~\ref{thm:minkvc-colorful}]
The proof is based on the ideas outlined above. For simplicity, we will describe the algorithm that computes an approximation for $\optmin(G, k, \chi)$ rather than a subset $S \subseteq V_G$, i.e., it will output a number between $\optmin(G, k, \chi)$ and $(1 + \varepsilon) \cdot \optmin(G, k, \chi)$. It is not hard to see that the algorithm can be turned to provide a desired set as well.

As stated above, the algorithm consists of two parts: the subgraph generation part, and the dynamic programming part. The subgraph generation algorithm, which is shown below as Algorithm~\ref{alg:subgraph-generation}, is very much the same as stated earlier: it takes as an input the sets $S_{\textsc{active}}$ and $S_{\textsc{included}}$ (in addition to $(G, k, \chi)$). If there is no more active vertex in $S_{\textsc{active}}$, then it just updates the table DP to reflect $E_G(S_{\textsc{included}})$. Otherwise, it pick a vertex $u$ and try to branch on every representative $T$ from $\cT$ from Lemma~\ref{lem:eps-net} where the $\{a_i\}$'s are defined as $a_v = w_{\{u, v\}}$ for all $v \ne \{u\}$ and $\delta = \varepsilon/2$.

\begin{algorithm}
\caption{}\label{alg:subgraph-generation}
\begin{algorithmic}[1]
\Procedure{SubgraphGen$(G, k, \chi, S_{\textsc{active}}, S_{\textsc{included}})$}{}
\If{$S_{\textsc{active}} = \emptyset$}
\State $\text{DP}[\chi(S_{\textsc{included}})] \leftarrow \min\{\text{DP}[\chi(S_{\textsc{included}})], E_G(S_{\textsc{included}})\}$
\Else
\State $u \leftarrow$ Any element of $S_{\textsc{active}}$
\State $S_{\textsc{active}} \leftarrow S_{\textsc{active}} \setminus \{u\}$
\State $\cT \leftarrow$ Subsets generated by Lemma~\ref{lem:eps-net} for $a_v = w_{\{u, v\}}$ for all $v \ne u$ and $\delta = \varepsilon / 2$.
\For{$T \subseteq \cT \cap \binom{V_G \setminus \{u\}}{\leqs k}$}
\If{$T \cap S_{\textsc{included}} = \emptyset$ \textbf{and} $S_{\textsc{included}} \cup T$ is colorful} \label{step:check}
\State {\sc SubgraphGen}$(G, k, \chi, S_{\textsc{active}} \cup T, S_{\textsc{included}} \cup T)$ \label{step:recurse}
\EndIf
\EndFor
\EndIf
\EndProcedure
\end{algorithmic}
\end{algorithm}

The dynamic programming (main algorithm) proceeds in a rather straightforward manner: after initializing the table, the main algorithm calls the subgraph generation subroutine starting with each vertex. Then, it uses dynamic programming to updates the table DP to reflect the fact that the answer may consist of many connected components. Finally, it outputs DP$[\{1, \dots, k\}]$. The pseudo-code for this is given below as Algorithm~\ref{alg:dp}.

\begin{algorithm}
\caption{}\label{alg:dp}
\begin{algorithmic}[1]
\Procedure{Min\_$k$-VC$(G, k, \chi)$}{}
\For{$C \subseteq [k]$}
\State $\text{DP}[C] \leftarrow \infty$
\EndFor
\For{$u \in V_G$}
\State \textsc{SubgraphGen}$(G, k, \chi, \{u\}, \{u\})$
\EndFor
\For{$C \subseteq [k]$ in increasing order of $|C|$} \label{step:dp}
\For{$C' \subseteq C$}
\State $\text{DP}[C] \leftarrow \min\{\text{DP}[C], \text{DP}[C'] + \text{DP}[C \setminus C']\}$
\EndFor
\EndFor
\State \Return $\text{DP}[[k]]$ 
\EndProcedure
\end{algorithmic}
\end{algorithm}

{\bf Running Time Analysis.} We will show that the running time of the algorithm is indeed $O(1/\varepsilon)^{O(k)}$. It is obvious that the dynamic programming step takes only $2^{O(k)} \cdot \poly(n)$ time, and it is not hard to see that each call to {\sc SubgraphGen}, without taking into account the time spent in the recursed calls (Step~\ref{step:recurse}), takes only $O(1/\varepsilon)^{O(k)} \cdot \poly(n)$ time (because the bottleneck is the generation of $\cT \cap \binom{V_G \setminus \{u\}}{\leqs k}$ and this takes only $O(1/\varepsilon)^{O(k)} \cdot \poly(n)$ time as guaranteed by Lemma~\ref{lem:eps-net}). Thus, it suffices for us to show that, for each $u \in V$, \textsc{SubgraphGen}$(G, k, \chi, \{u\}, \{u\})$ only generates $O(1/\varepsilon)^{O(k)} \cdot \poly(n)$ leaves in the recursion tree. (By \emph{leaves}, we refer to calls {\sc SubgraphGen}$(G, k, \chi, S_{\textsc{active}}, S_{\textsc{included}})$ where $S_{\textsc{active}} = \emptyset$. Note that, if {\sc SubgraphGen}$(G, k, \chi, \emptyset, S_{\textsc{included}})$ is called multiple times for the same $S_{\textsc{included}}$, we count each call separately.) The proof is a formalization of the ``branching factor per vertex added'' idea outlined before the proof.

In fact, we will prove an even more general statement: for all colorful subsets $S_{\textsc{active}} \subseteq S_{\textsc{included}}$, {\sc SubgraphGen}$(G, k, \chi, S_{\textsc{active}}, S_{\textsc{included}})$ results in only at most $(C/\varepsilon)^{2k - |S_{\textsc{included}}| - |S_{\textsc{included}}\setminus S_{\textsc{active}}|}$ leaves for some $C > 0$. In particular, let $C' > 0$ be a constant such that Lemma~\ref{lem:eps-net} gives the bound $|\cT \cap \binom{[\ell]}{j}| \leqs (C'/\delta)^j$; we will prove the statement for $C = 2C' + 2$.

We prove by induction on decreasing order of $|S_{\textsc{included}}|$ and $|S_{\textsc{included}}\setminus S_{\textsc{active}}|$ respectively. In the base case where $|S_{\textsc{included}}| = k$, the statement is obviously true, since the condition in Line~\ref{step:check} ensures that no more subroutine is executed. In another base case where $|S_{\textsc{included}}\setminus S_{\textsc{active}}| = |S_{\textsc{included}}|$, the statement is also obviously true since, in this case, we simply have $S_{\textsc{active}} = \emptyset$.

For the inductive step, suppose that, for some $0 \leqs i < k$ and $1 \leqs j \leqs i$, the statement holds for all colorful subsets $S_{\textsc{active}} \subseteq S_{\textsc{included}}$such that $|S_{\textsc{included}}| > i$, or, $|S_{\textsc{included}}| = i$ and $|S_{\textsc{active}}| < j$. Now, consider any colorful subsets $S_{\textsc{active}} \subseteq S_{\textsc{included}}$ such that $|S_{\textsc{included}}| = i$ and $|S_{\textsc{active}}| = j$. We will argue below that {\sc SubgraphGen}$(G, k, \chi, S_{\textsc{active}}, S_{\textsc{included}})$ results in at most $(C'/\varepsilon)^{2k - i - (i - j)}$ leaves.

To do so, first observe that (1) in every recursive call, $|S_{\textsc{included}}\setminus S_{\textsc{active}}|$ increases by one (namely $u$ becomes inactive) and (2) for every $0 \leqs t \leqs k - i$,  the number of recursive calls for which $|S_{\textsc{included}}|$ increases by $t$ is at most $|\cT \cap \binom{V_G \setminus \{u\}}{t}| \leqs (C'/\varepsilon)^t$. As a result, by the inductive hypothesis, the number of leaves generated by {\sc SubgraphGen}$(G, k, \chi, S_{\textsc{active}}, S_{\textsc{included}})$ is at most
\begin{align*}
\sum_{t=0}^{k - i} (C'/\varepsilon)^t \cdot (C/\varepsilon)^{2k - (i + t) - (i - j + 1)} 
&= (C/\varepsilon)^{2k - i - (i - j + 1)} \cdot \left(\sum_{t=0}^{k - i} (C'/C)^t\right) \\
(\text{Since } C \geqs 2C') &\leqs (C/\varepsilon)^{2k - i - (i - j + 1)} \cdot 2 \\
(\text{Since } C \geqs 2) &\leqs (C/\varepsilon)^{2k - i - (i - j)}
\end{align*}
as desired. 

In conclusion, for all colorful subsets $S_{\textsc{active}} \subseteq S_{\textsc{included}}$, {\sc SubgraphGen}$(G, k, \chi, S_{\textsc{active}}, S_{\textsc{included}})$ generates at most $(C/\varepsilon)^{2k - |S_{\textsc{included}}| - |S_{\textsc{included}}\setminus S_{\textsc{active}}|}$ leaves. As argued above, this implies that the running time of the algorithm is at most $O(1/\varepsilon)^{O(k)} \cdot \poly(n)$.

{\bf Approximation Guarantee Analysis.} We will now show that the output lies between $\opt_{\text{Min }k\text{-VC}}(G, k, \chi)$ and $(1 + \varepsilon) \cdot \opt_{\text{Min }k\text{-VC}}(G, k, \chi)$. For convenience, let us define $\text{DP}^*$ to be the value of table $\text{DP}$ filled by {\sc SubgraphGen} calls; that is, this is the table before Line~\ref{step:dp} in Algorithm~\ref{alg:dp}. Observe the following relationship between DP and DP$^*$:
\begin{align} \label{eq:obs-partition}
\text{DP}[C] = \min_{\text{Partition } P \text{ of } C} \sum_{C' \in P} \text{DP}^*[C'].
\end{align}
It is now rather simple to see that the output is at least $\opt_{\text{Min }k\text{-VC}}(G, k, \chi)$. To do so, observe that, for any $C \subseteq [k]$, \text{DP}$^*[C]$ is equal $E_G(S_C)$ for some colorful $S_C \subseteq V_G$ with $\chi(S_C) = C$. This, together with~\eqref{eq:obs-partition}, implies that the output must be equal to $\sum_{C' \in P} E_G(S_{C'})$ for some partition $P$ of $[k]$ and colorful $S_{C'}$'s such that $\chi(S_{C'}) = C'$. Observe that this value is at least $E_G\left(\bigcup_{C' \in P} S_{C'}\right)$, which is at least $\opt_{\text{Min }k\text{-VC}}(G, k, \chi)$ since $\bigcup_{C' \in P} S_{C'}$ is a colorful set of size $k$.

Next, we will show that the output (i.e. DP$[[k]]$) is at most $(1 + \varepsilon) \cdot \opt_{\text{Min }k\text{-VC}}(G, k, \chi)$. The following proposition is at the heart of this proof:

\begin{proposition} \label{prop:recurse}
For any non-empty colorful subset $S \subseteq V_G$, there exists a non-empty $S^{\rep} \subseteq S$ such that 
\begin{align*}
\text{DP}^*[\chi(S^{\rep})] \leqs E_G(S^{\rep}) \text{ and } E_G(S^{\rep}, S \setminus S^{\rep}) \leqs  \delta \cdot \wdeg(S^{\rep}). 
\end{align*}
\end{proposition}

\begin{subproof}[Proof of Proposition~\ref{prop:recurse}]
Let $v$ be any vertex in $S$. Let us consider the call {\sc SubgraphGen}$(G, \chi, k, \{v\}, \{v\})$. Consider traversing the following single branch in every execution of Step~\ref{step:recurse}: pick $T \in \cT$ such that $T \subseteq (S \setminus S_{\textsc{included}})$ and $\sum_{i \in (S \setminus S_{\textsc{included}}) \setminus T} w_{\{u, i\}} \leqs \delta \cdot \sum_{i \in V_G} w_{\{u, i\}} = \delta \cdot \wdeg_G(u)$. (We remark that such $T$ is guaranteed to exist by Lemma~\ref{lem:eps-net}; if there are more than one such $T$'s, just choose an arbitrary one.) Suppose that always choosing such branch ends in a call {\sc SubgraphGen}$(G, k, \chi, \emptyset, S^{\rep})$. We will show that $S^{\rep}$ satisfies the desired properties.

First of all, observe that the fact we always choose $T \subseteq S$ ensures that $S^{\rep} \subseteq S$ and that, since {\sc SubgraphGen}$(G, k, \chi, \emptyset, S^{\rep})$ is executed, we indeed have DP$[\chi(S^{\rep})] \leqs E_G(S^{\rep})$. Hence, we are only left to argue that $E_G(S^{\rep}, S \setminus S^{\rep}) \leqs \delta \cdot \wdeg(S^{\rep})$. To see that this is the case, observe that the second property of the $T$'s chosen implies that $\sum_{i \in S \setminus S^{\rep}} w_{\{u, i\}} \leqs \delta \cdot \wdeg(u)$. Summing this inequality over all $u \in S^{\rep}$ immediately yields $E_G(S^{\rep}, S \setminus S^{\rep}) \leqs \delta \cdot \wdeg(S^{\rep})$.
\end{subproof}

With Proposition~\ref{prop:recurse} ready, we can now prove that DP$[[k]] \leqs (1 + \varepsilon) \cdot \opt_{\text{Min }k\text{-VC}}(G, k, \chi)$. Let $S_{\opt} \subseteq V_G$ denote an optimal solution to the problem, i.e., $S_{\opt}$ is a colorful $k$-vertex subset such that $E_G(S_{\opt}) = \optmin(G, k, \chi)$. Let $S_1 = S_{\opt}$. For $i = 1, \dots$, if $S_i \ne \emptyset$, we apply Proposition~\ref{prop:recurse} to find a non-empty subset $S^{\rep}_i \subseteq S_i$ such that 
\begin{align} \label{eq:recurse}
\text{DP}^*[\chi(S^{\rep}_i)] \leqs E_G(S^{\rep}_i) \text{ and } E_G(S^{\rep}_i, S_{i + 1}) \leqs  \delta \cdot \wdeg(S^{\rep}_i). 
\end{align}
where $S_{i + 1} = S_i \setminus S_i^{\rep}$.

Observe here that $\{S^{\rep}_i\}_{i \geqs 1}$ is a partition of $S_{\opt}$. Thus, from~\eqref{eq:obs-partition} and~\eqref{eq:recurse}, we have 
\begin{align} \label{eq:dp-upper}
\text{DP}[[k]] \stackrel{\eqref{eq:obs-partition}}{\leqs} \sum_{i \geqs 1} \text{DP}^*[\chi(S^{\rep}_i)] \stackrel{\eqref{eq:recurse}}{\leqs} \sum_{i \geqs 1} E_G(S^{\rep}_i). 
\end{align}

On the other hand, observe that $E_G(S_i) = E_G(S^{\rep}_i) + E_G(S_{i + 1}) - E_G(S^{\rep}_i, S_{i + 1})$. Thus, we have
\begin{align}
E_G(S_{\opt}) = &\sum_{i \geqs 1} (E_G(S_i) - E_G(S_{i + 1})) \nonumber \\
&= \sum_{i \geqs 1} E_G(S^{\rep}_i) - \sum_{i \geqs 1} E_G(S^{\rep}_i, S_{i + 1}) \nonumber \\
&\stackrel{\eqref{eq:recurse}}{\geqs} \sum_{i \geqs 1} E_G(S^{\rep}_i) - \delta \cdot \sum_{i \geqs 1} \wdeg(S^{\rep}_i) \nonumber \\
&= \sum_{i \geqs 1} E_G(S^{\rep}_i) - \delta \cdot \wdeg(S_{\opt}). \label{eq:sopt-bound}
\end{align}
Finally, from~\eqref{eq:dp-upper},~\eqref{eq:sopt-bound} and $\wdeg(S_{\opt}) \leqs 2 \cdot E_G(S_{\opt})$, we have DP$[[k]] \leqs (1 + 2\delta) \cdot E_G(S_{\opt}) = (1 + \varepsilon) \cdot \opt_{\text{Min }k\text{-VC}}(G, k, \chi)$ which concludes the proof.
\end{proof}

\subsection{Non-Existence of Polynomial Size Approximate Kernel} \label{sec:no-kernel}

The above FPT-AS and the equivalence between existence of FPT approximation algorithm and approximate kernel~\cite{LPRS17} immediately implies that there exists an $(1 - \varepsilon)$-approximate kernel for Weighted Min $k$-VC. However, this naive approach results in an approximate kernel of size $O(1/\varepsilon)^{O(k)}$. A natural question is whether there exists a polynomial-size approximate kernel for Weighted Min $k$-VC (similar to Weighted Max $k$-VC). In this section, we show that the answer to this question is likely a negative, assuming a variant of the Small Set Expansion Conjecture.

Our proof follows the framework of Lokshtanov \etal~\cite{LPRS17}. Let us recall that an equivalence relation $R$ over strings on a finite alphabet $\Sigma$ is said to be \emph{polynomial} if (i) whether $x \sim y$ can be checked in $\poly(|x| + |y|)$ time and (ii) for every $n \in \N$, $\Sigma^n$ has at most $\poly(n)$ equivalence classes. The framework of Lokshtanov \etal\ uses the notion of $\alpha$-gap cross composition, as defined below. (This is based on the cross composition in the exact settings from~\cite{BodlaenderJK14}.)

\begin{definition}[$\alpha$-gap cross composition~\cite{LPRS17}] Let $L$ be a language and $\Pi$ be a parameterized minimization problem. We say that $L$ \emph{$\alpha$-gap
cross composes} into $\Pi$ (for $\alpha \leqs 1$), if there is a polynomial equivalence relation $R$ and an
algorithm which, given strings $x_1, \cdots, x_t$ from the same equivalence class of $R$, computes
an instance $(y, k)$ of $\Pi$ and $r \in \mathbb{R}$, in time $\poly(\sum_{i=1}^t |x_i|)$ such that the following holds:
\begin{enumerate}[(i)]
\item (Completeness) $\opt_{\Pi}(y, k) \leqs r$ if $x_i \in L$ for some $1 \leqs i \leqs t$,
\item (Soundness) $\opt_{\Pi}(y, k) > r\alpha$ if $x_i \notin L$ for all $i \in [t]$, and,
\item $k$ is bounded by a polynomial in $\log t + \max_{1 \leqs i \leqs t} |x_i|$.
\end{enumerate}
\end{definition}

A parameterized optimization problem is said to be \emph{nice} if, given a solution to the problem, its cost can be computed in polynomial time. (Clearly, Weighted Min $k$-VC is nice.) The main tool from~\cite{LPRS17} is that any problem that $\alpha$-gap cross composes to a nice parameterized optimization problem $\Pi$ must be in $\coNP/\poly$ if $\Pi$ has $\alpha$-approximate kernel\footnote{We note that the result of~\cite{LPRS17} works even with a weaker notion than $\alpha$-approximate kernel called \emph{$\alpha$-approximate compression}; see Definition 5.5 and Theorem 5.9 of~\cite{LPRS17} for more details.}. In other words, if an NP-hard language $\alpha$-gap cross composes to $\Pi$, then $\Pi$ does not have $\alpha$-approximate kernel unless $\NP \subseteq \coNP/\poly$.

\begin{lemma}[\cite{LPRS17}] \label{lem:compose-lowerbound}
Let $L$ be a language and $\Pi$ be a nice parameterized optimization problem. If $L$ $\alpha$-gap cross composes to $\Pi$, and $\Pi$ has a polynomial size $\alpha$-approximate kernel, then $L \in \coNP/\poly$.
\end{lemma}

As stated earlier, our lower bound will be based on the Small Set Expansion Hypothesis (SSEH)~\cite{RS10}. To state the hypothesis, let us first recall the definition of edge expansion; for a graph $G$, the \emph{edge expansion} of a subset of vertices $S \subseteq V_G$ is defined as $\Phi(S) := \frac{E_G(S, V_G \setminus S)}{\wdeg(S)}$. Roughly speaking, SSEH, which was proposed in~\cite{RS10}, states that it is NP-hard to determine whether (completeness) there is a subset of a specified size with very small edge expansion or (soundness) every subset of a specified size has edge expansion close to one. This is formalized below.

\begin{definition}[SSE($\delta, \eta$)]
Given an unweighted regular graph $G$, distinguish between:\begin{itemize}
\item (Completeness) There exists $S \subseteq V_G$ of size $\delta|V_G|$ such that $\Phi(S) \leqs \eta$.
\item (Soundness) For every $S \subseteq V_G$ of size $\delta|V_G|$, $\Phi(S) > 1 - \eta$.
\end{itemize}
\end{definition}

\begin{conjecture}[Small Set Expansion Hypothesis~\cite{RS10}] \label{conj:sseh}
For every $\eta > 0$, there exists $\delta = \delta(\eta) > 0$ such that SSE$(\delta, \eta)$ is NP-hard. 
\end{conjecture}

Before we state the variant of SSEH that we will use, let us demonstrate why we need to strengthen the hypothesis. To do so, let us consider the $(2 - \varepsilon)$-factor hardness of approximation of Min $k$-VC as proved in~\cite{GandhiK15}, which our construction will be based on. The reduction takes in an input $G$ to SSE($\delta, \eta$) and simply just outputs $(G, k)$ where $k = \delta |V_G|$. The point is that, in a $d$-regular graph, a set $S$ covers exactly $d(1 + \Phi(S))|S|/2$ edges. This means that, in the completeness case, there is a set $S$ of size $k$ that covers only $d(1 + \eta)k/2$ edges, whereas, in the soundness case, any set $S$ of size $k$ covers at least $d(2 - \eta)k/2$ edges. By selecting $\eta$ sufficiently small, the ratio between the two cases is at least $(2 - \varepsilon)$, and hence~\cite{GandhiK15} arrives at their $(2 - \varepsilon)$-factor inapproximability result.

Now, our cross composition is similar to this, except that we need to be to handle multiple instances at once. More specifically, given instance $G_1, \dots, G_t$ of SSE($\delta, \eta$) where all $G_1, \dots, G_t$ are $d$-regular for some $d$ and $|V_{G_1}| = \cdots = |V_{G_t}|$, we want to produce an instance $(G^*, k)$ where $G^*$ is the disjoint union of $G_1, \dots, G_t$ and $k = \delta |V|$. Once again, the completeness case works exactly as before. The issue lies in the soundness case: even though we know that every $S_i \subseteq V_{G_i}$ of size $k$ has expansion close to one, it is possible that there exists $S_i \subseteq V_{G_i}$ of size much smaller than $k$ that has small expansion. For instance, it might even be that $G_1, \dots, G_t$ each contains a connected component of size $k/t$. In this case, we can take the union of these components and arrive at a set of size $k$ that covers $dk/2$ edges, which is even smaller than the completeness case! In other words, for the composition to work, we want the soundness of SSEH to consider not only $S$'s of size $k$, but also $S$'s of size \emph{at most} $k$. With this in mind, we can formalize our strengthened hypothesis as follows.

\begin{definition}[Strong-SSE($\delta, \eta$)]
Given an unweighted regular graph $G$, distinguish between:\begin{itemize}
\item (Completeness) There exists $S \subseteq V_G$ of size $\delta|V_G|$ such that $\Phi(S) \leqs \eta$.
\item (Soundness) For every $S \subseteq V_G$ of size at most $\delta|V_G|$, $\Phi(S) > 1 - \eta$.
\end{itemize}
\end{definition}

\begin{conjecture}[Strong Small Set Expansion Hypothesis] \label{conj:strong-sseh}
For every $\eta > 0$, there exists $\delta = \delta(\eta) > 0$ such that Strong-SSE$(\delta, \eta)$ is NP-hard. 
\end{conjecture}

We remark that it is known that a strengthening of SSEH where the soundness case is required for all $S$ of size in $[\beta\delta|V|, \delta|V|]$ for any $\beta > 0$ is known to be equivalent to the original SSEH. (See Appendix A.2 of the full version of~\cite{RST12} for a simple proof.) This is closely related to what we want above, except that we need this to holds even for $|S| = o(|V|)$. To the best of our knowledge, the Strong SSEH as stated above is not known to be equivalent to the original SSEH.

\begin{proof}[Proof of Lemma~\ref{lem:no-kernel}]
Let $\varepsilon$ be any number that lies in $(0, 1]$. Let $\eta$ be $\varepsilon/2$, and let $\delta = \delta(\eta) > 0$ be as guaranteed by Conjecture~\ref{conj:sseh}. We will show that Strong-SSE$(\delta, \eta)$ $(2 - \varepsilon)$-gap cross composes\footnote{Note that strictly speaking Strong-SSE$(\delta, \eta)$ is not a language, but rather a promise problem (cf.~\cite{Goldreich06a}). Nevertheless, the notion of gap cross composes extends naturally to promise problems; the only changes are that in the yes case $x_i \in L$ should be changed to $x_i \in L_{\text{YES}}$ and in the no case $x_i \notin L$ should be changed to $x_i \in L_{\text{NO}}$. The result in Lemma~\ref{lem:compose-lowerbound} also holds for this case; for instance, see Lemma 5.11 and Theorem 5.12 of~\cite{LPRS17}, where the gap cross composition also starts from a promise problem (Gap-Longest-Path).} into Min $k$-VC, which together with Lemma~\ref{lem:compose-lowerbound} immediately implies the statement in the lemma.

We define an equivalence relation $R$ on instances of Strong-SSE$(\delta, \eta)$ by $G \sim G'$ iff $|V_G| = |V_{G'}|$ and $\wdeg(G) = \wdeg(G')$. It is obvious that $R$ is polynomial. Given $t$ instances $G_1, \dots, G_t$ from the same equivalence class of $R$ where $n = |V_{G_1}| = \cdots = |V_{G_t}|$ and $d = \wdeg(G_1) = \cdots = \wdeg(G_t)$, we create an instance $(G^*, k)$ of Min $k$-VC by letting $G^*$ be the (disjoint) union of $G_1, \dots, G_t$, $k = \delta n$, and $r = d\delta n(1 + \eta)/2$. We next argue the completeness and soundness of the composition.

{\bf Completeness.} Suppose that, for some $i \in [t]$, there exists $S \subseteq V_{G_i}$ of size $\delta n$ such that $\Phi_{G_i}(S) \leqs \eta$. Then, the number of edges covered by $S$ (in both $G_i$ and $G^*$) is $d \delta n(1 + \Phi(S))/2 \leqs d \delta n(1 + \eta)/2$. In other words, $\opt_{\text{Max }k\text{-VC}}(G^*, k) \leqs r$ as desired. 

{\bf Soundness.} Suppose that, for all $i \in [t]$ and $S \subseteq V_{G_i}$ of size at most $\delta n$, we have $\Phi_{G_i}(S) > (1 - \eta)$. Consider any set $S^* \subseteq V_{G^*}$ of size $\delta n$. Let $S_i$ denote $S^* \cap V_{G_i}$. Observe that the number of edges covered by $S^*$ is
\begin{align*}
\sum_{i \in [t]} d|S_i|(1 + \Phi_{G_i}(S_i))/2 \geqs \sum_{i \in [t]} d|S_i|(2 - \eta)/2 = d \delta n(2 - \eta)/2 \geqs (2 - \varepsilon) r,
\end{align*}
where the first inequality comes from our assumption and the second comes from our choice of $\eta$. Thus, we have $\opt_{\text{Max }k\text{-VC}}(G^*, k) > (2 - \varepsilon)r$ as desired.
\end{proof}

We note here that the above proof produces $G^*$ that is unweighted. As a result, the lower bound also applies for Unweighted Min $k$-VC.

\section{Concluding Remarks}

Let us make a few brief remarks regarding the tightness of running times of our algorithms. 
\begin{itemize}
\item The \W[1]-hardness proofs of Max $k$-VC and Min $k$-VC in~\cite{GuoNW07} also implies that, even in the unweighted case, if we can approximate the problems to within $(1 - 1/n^2)$ and $(1 + 1/n^2)$ factors respectively, then we can solve the $k$-Clique problem with only polynomial overhead in running time. This implies the following lower bounds:
\begin{enumerate}
\item Unless $\W[1] = \FPT$, there is no FPT-AS for Max $k$-VC and Min $k$-VC with running time $\exp(f(k) \cdot o(\log(1/\varepsilon))) \cdot \poly(n)$ for any function $f$ (because this would give an FPT time algorithm for $k$-Clique when plugging in $\varepsilon = 1/n^2$).
\item Unless $k$-Clique can be solved in $g(k) \cdot n^{o(k)}$ time for some function $g$, there is no FPT-AS for Max $k$-VC and Min $k$-VC with running time $O(1/\varepsilon)^{o(k)} \cdot \poly(n)$. Note that such lower bound for $k$-Clique holds under the Exponential Time Hypothesis (ETH)\footnote{ETH states that no $2^{o(n)}$-time algorithm can solve 3SAT~\cite{IP01,IPZ01}.}~\cite{ChenHKX06}.
\end{enumerate}
\item For Max $k$-VC, the reduction that proves $(1 + \delta)$-factor NP-hardness of approximation~\cite{Patrank94} is in fact a linear size reduction from the gap version of 3SAT. As a result, assuming the Gap Exponential Time Hypothesis (Gap-ETH)\footnote{Gap-ETH states that there is no $2^{o(n)}$-time algorithm that can distinguish between a fully satisfiable 3CNF formula and one which is not even $0.999$-satisfiable~\cite{Dinur16,MR17}.}, there is no FPT-AS that runs in time $f(1/\varepsilon)^{o(k)} \cdot \poly(n)$ for any function $f$. Under the weaker ETH, a lower bound of the form $f(1/\varepsilon)^{o(k/\poly\log k)} \cdot \poly(n)$ for any $f$ can be achieved via nearly linear size PCP~\cite{Dinur07}. (Note that we do not know any lower bound of this form for Min $k$-VC; in particular, it is not known whether Min $k$-VC is NP-hard to approximate even for a factor of 1.0001.)
\end{itemize}

An interesting remaining open question is to close the gap between the (polynomial time) approximation algorithms and hardness of approximation for Max $k$-VC. On the algorithmic front, we note that Austrin \etal~\cite{ABG13} further exploited the techniques developed by Raghavendra and Tan~\cite{RT12} to achieve several improvements. Most importantly, they show that, for Max 2SAT with cardinality constraint, if the cardinality constraint is $x_1 + \cdots + x_n = n/2$ (i.e. $k = n/2$), then an $0.94$-approximation can be achieved in polynomial time. (In particular, the ratio here is the same as the ratio of the Lewin-Livnat-Zwick algorithm for Max 2SAT without cardinality constraint~\cite{LewinLZ02}; see also~\cite{Aus07,Sjo09}. Note that this ratio is still different from the hardness from~\cite{AKS11}.) This specific case is often referred to as \emph{Max Bisection 2SAT}. Unfortunately, the algorithm does not naturally\footnote{In particular, the rounding algorithm involves scaling the bias of the variables (see Section 6 of~\cite{ABG13}). For \emph{Max Bisection 2SAT}, the sum of the bias is zero and hence scaling retains the sum. However, when the sum is non-zero, scaling changes the sum and hence the rounding algorithm produces a subset of size not equal to $k$.} extend to the case where $k \ne n/2$ and hence it is unclear how to employ this algorithm for Max $k$-VC.

On the hardness of approximation front, we remark that the hardness that follows from~\cite{AKS11} holds even for the \emph{perfect completeness} case. That is, even when we are promised that there is a vertex cover of size $k$, it is still hard to find $k$ vertices that cover $0.944$ fraction of the edges. (See Appendix~\ref{app:aks}.) Interestingly, there is an evidence that this perfect completeness case is easier: Feige and Langberg~\cite{FL01} shows that their algorithm achieves 0.8-approximation in this case, which is better than $(0.75 + \delta)$-approximation that their algorithm yields in the general case. In fact, we can even get 0.94-approximation in this case as follows. First, we follow the kernelization for Vertex Cover~\cite{ChenKJ01} based on the Nemhauser-Trotter theorem~\cite{NemhauserT74}: on input graph $(G, k)$, this gives a partition $V_0, V_{1/2}, V_1$ such that there exists a vertex cover $S$ of size $k$ such that $V_1 \subseteq S \subseteq V_{1/2} \cup V_1$. Moreover, the Nemhauser-Trotter theorem also ensures that $|V_{1/2}| = 2 \cdot (k - |V_1|)$. This means that we can restrict ourselves to the graph induced by $V_{1/2}$ and applies the aforementioned Max Bisection 2SAT from~\cite{ABG13}. This indeed gives us a $0.94$-approximation as desired. These suggest that it might be that the perfect completeness case is easier to approximate; thus, it would be interesting to see whether there is any way to construct harder instances with imperfect completeness.

\bibliography{main}
\bibliographystyle{alpha}

\appendix

\section{Inapproximability of Max $k$-VC from~\cite{AKS11}} \label{app:aks}

In this section, we briefly sketch how Austrin, Khot and Safra's proof of inapproximability of Vertex Cover and Independent Set in bounded degree graphs~\cite{AKS11} immediately implies a 0.944-factor hardness of approximation for Max $k$-VC. Note that we decide to include this since it does not seem to appear anywhere yet.

Let $\Phi$ denote the cumulative density function of the standard normal
distribution and, for any $\rho \in [-1, 1], \mu \in [0, 1]$, let $\Gamma_{\rho}(\mu)$ denote $\Pr[X \leqs \Phi^{-1}(\mu) \wedge Y \leqs \Phi^{-1}(\mu)]$ where $X, Y$ are normal random variables with means 0, variances 1 and covariance $\rho$. The main intermediate result of~\cite{AKS11} is the following:

\begin{theorem}[Theorem 1 from \cite{AKS11}]
For any $q \in (0, 1/2)$ and any $\varepsilon > 0$, it is UG-hard to, given a graph $G = (V_G, E_G)$, distinguish between the following two cases.
\begin{itemize}
\item (Completeness) $G$ contains an independent set of size at least $q \cdot |V_G|$.
\item (Soundness) For any subset $T \subseteq V_G$, the number of edges with both endpoint in $T$ is at least $|E_G| \cdot \left(\Gamma_{-q/(1 - q)}(\mu) - \varepsilon\right)$ where $\mu = |T|/|V_G|$.
\end{itemize}
\end{theorem}

This means that, in the completeness case, there is a vertex cover of size at most $(1 - q) \cdot |V_G|$. On the other hand, in the soundness case, if we consider any subset $S \subseteq V_G$ of size at most $(1 - q) \cdot |V_G|$, then the number of edges \emph{not} covered is exactly the same as the number of edges with both endpoints in $(V_G \setminus S)$, which is at least $(\Gamma_{-q/(1 - q)}(q) - \varepsilon) \cdot |E_G|$. In other words, for any $q \in (0, 1/2)$, it is UG-hard to approximate Max $k$-VC to within a factor of $(1 - \Gamma_{-q/(1 - q)}(q) + \varepsilon)$ for any $\varepsilon > 0$. That is, Austrin \etal's result implies the following:

\begin{corollary}
Let $\alpha_{\text{AKS}} = \sup_{q \in (0, 1/2)} (1 - \Gamma_{-q/(1 - q)}(q))$. For every $\varepsilon > 0$, it is UG-hard to approximate Max $k$-VC to within a factor of $\alpha_{\text{AKS}} + \varepsilon$.
\end{corollary}

Numerically, $\alpha_{\text{AKS}}$ lies between $0.943$ and $0.944$. Thus, it is UG-hard to approximate Max $k$-VC to within 0.944 factor.

\end{document}